\documentclass[letterpaper,twocolumn,twoside]{IEEEtran} 
\usepackage{amsmath, amssymb, bm, cite, epsfig, psfrag,amsthm}
\usepackage{algorithm,algorithmic}
\usepackage{bbm}
\usepackage[usenames]{color}
\usepackage{tikz}
\usepackage{etoolbox}
\usepackage{enumerate}
\usetikzlibrary{shapes,arrows}
\usetikzlibrary{patterns}
\usepackage{pgfplots}
\usepackage{url,hyperref}
\usepackage{xspace}
\usepackage{enumerate}
\usepackage{mathtools}

\newtoggle{conference}
\togglefalse{conference} 

\usepackage[normalem]{ulem} 

\def\beq{\begin{equation}}
\def\eeq{\end{equation}}
\def\beqa{\begin{eqnarray}}
\def\eeqa{\end{eqnarray}}
\def\beqan{\begin{eqnarray*}}
\def\eeqan{\end{eqnarray*}}

\def\R{{\mathbb{R}}}

\def\argmin{\mathop{\mathrm{arg\,min}}}
\def\argmax{\mathop{\mathrm{arg\,max}}}
\def\diag{\mathop{\mathrm{diag}}}
\def\Diag{\mathop{\mathrm{Diag}}}
\def\x{\times}

\newtheorem{theorem}{Theorem}
\newtheorem{lemma}{Lemma}

\newtheorem{assumption}{Assumption}

\def\xhat{\widehat{x}}

\def\arr{\rightarrow}
\def\Exp{\mathbb{E}}
\def\var{\mathop{\mathrm{var}}}

\def\Cov{\mathop{\mathrm{Cov}}}
\def\Tr{\mathop{\mathrm{tr}}}
\def\diag{\mbox{diag}}

\def\tm1{t\! - \! 1}
\def\tp1{t\! + \! 1}
\def\km1{k\! - \! 1}
\def\kp1{k\! + \! 1}

\newcommand{\bbf}{\mathbf{b}}

\newcommand{\dbf}{\mathbf{d}}
\newcommand{\ebf}{\mathbf{e}}
\newcommand{\gbf}{\mathbf{g}}

\newcommand{\pbf}{\mathbf{p}}

\newcommand{\qbf}{\mathbf{q}}

\newcommand{\rbf}{\mathbf{r}}

\newcommand{\sbf}{\mathbf{s}}

\newcommand{\ubf}{\mathbf{u}}

\newcommand{\vbf}{\mathbf{v}}

\newcommand{\wbf}{\mathbf{w}}

\newcommand{\xbf}{\mathbf{x}}
\newcommand{\xbfhat}{\widehat{\mathbf{x}}}

\newcommand{\ybf}{\mathbf{y}}
\newcommand{\zbf}{\mathbf{z}}

\newcommand{\zbfhat}{\widehat{\mathbf{z}}}
\newcommand{\Abf}{\mathbf{A}}

\newcommand{\Ibf}{\mathbf{I}}
\newcommand{\Jbf}{\mathbf{J}}

\newcommand{\Pbf}{\mathbf{P}}

\newcommand{\Qbftilde}{\widetilde{\mathbf{Q}}}
\newcommand{\Qbf}{\mathbf{Q}}

\newcommand{\Sbf}{\mathbf{S}}

\newcommand{\Ybf}{\mathbf{Y}}

\def\betabf{{\boldsymbol \beta}}

\def\gammabf{{\boldsymbol \gamma}}
\def\Gammabf{{\boldsymbol \Gamma}}

\def\etabf{{\boldsymbol \eta}}
\def\etabfhat{\widehat{\boldsymbol \eta}}

\newcommand{\tran}{^{\text{\sf T}}}

\def\gt{\widetilde{\gbf}}

\newcommand*\dif{\mathop{}\!\mathrm{d}} 
\newcommand{\Hessian}[1]{\boldsymbol{\mathcal{H}}_{#1}} 
\newcommand{\vmult}{.}
\newcommand{\vdiv}{./}

\tikzstyle{block}=[rectangle,draw, fill=blue!20,
    minimum height=1cm, minimum width=1.5cm]
\tikzstyle{signal}=[coordinate,draw]

\title{Expectation Consistent Approximate Inference:  Generalizations and Convergence}

\iftoggle{conference}{
   \author{
     \IEEEauthorblockN{
        Alyson Fletcher,\IEEEauthorrefmark{1}
        Mojtaba Sahraee-Ardakan,\IEEEauthorrefmark{2}
        Sundeep Rangan,\IEEEauthorrefmark{3}
        and
        Philip Schniter\IEEEauthorrefmark{4}
     }
      \IEEEauthorblockA{
        \IEEEauthorrefmark{1}UCLA, akfletcher@ucla.edu,
        \IEEEauthorrefmark{2}UCSC, msahraee@ucsc.edu
     }
     \IEEEauthorblockA{
        \IEEEauthorrefmark{3}NYU, srangan@nyu.edu,
         \IEEEauthorrefmark{4}The Ohio State Univ., schniter@ece.osu.edu
     }
   }
}{
  \author{
    Alyson K. Fletcher, \IEEEmembership{Member,~IEEE},
    Mojtaba Sahraee-Ardakan, \IEEEmembership{Student Member,~IEEE},\\
    Sundeep Rangan, \IEEEmembership{Fellow,~IEEE},
    and
    Philip Schniter, \IEEEmembership{Fellow,~IEEE}
   \thanks{A.~Fletcher and M. Saharee (email: akfletcher@ucla.edu,
       msahraee@ucsc.edu) are with
       the Department of Statistics and Electrical Engineering,
       University of California, Los Angeles.  Their work
       is supported in part by the National Science
       Foundation under Grant No.~1254204 and the Office of Naval Research
       Grant No.~N00014-15-1-2677.}
     \thanks{S. Rangan (email: srangan@nyu.edu) is with
           the Department of Electrical and Computer Engineering,
           New York University, Brooklyn, NY.
       His work was supported by the National Science
       Foundation under Grant No. 1116589 and the industrial affiliates of NYU
       WIRELESS.}
     \thanks{P.~Schniter (email: schniter@ece.osu.edu) is with
           the Department of Electrical and Computer Engineering,
           The Ohio State University.
       His work was supported in part by
       the National Science Foundation under Grants CCF-1218754 and CCF-1527162.}
   }
}

\begin{document}
\setlength{\arraycolsep}{0.8mm}

\maketitle
\begin{abstract}
Approximations of loopy belief propagation, including expectation propagation
and approximate message passing, have attracted considerable attention for
probabilistic inference problems.  This paper proposes and analyzes
a generalization of Opper and Winther's
expectation consistent (EC) approximate inference method.
The proposed method, called Generalized Expectation Consistency (GEC),
can be applied to both maximum a posteriori (MAP) and minimum mean squared error (MMSE)
estimation.
Here we characterize its fixed points, convergence,
and performance relative to the replica prediction of optimality.\!
\end{abstract}


\section{Introduction}

Consider the problem of estimating a random vector $\xbf \in \R^N$ from
observations $\ybf\in\R^M$ under the posterior density
\beq \label{eq:pxy}
    p(\xbf|\ybf) = Z^{-1} \exp\left[ -f_1(\xbf) -f_2(\xbf) \right],
\eeq
where
$Z = \int \exp\left[ -f_1(\xbf) -f_2(\xbf) \right] \dif\xbf$
is a normalization constant%
\iftoggle{conference}{}{sometimes called the \emph{partition function}}
and $f_i(\xbf)$ are penalty functions.
Although both $Z$ and the penalties $f_i$ may depend on $\ybf$,
our notation suppresses this dependence.
We are interested in two problems:
\begin{itemize}
\item \textbf{MAP estimation}:
\iftoggle{conference}{Here}{In \emph{maximum a posteriori} (MAP) estimation, }
we wish to find the point estimate
$\xbfhat = \argmax_\xbf p(\xbf|\ybf)$,
equivalently stated as
\beq \label{eq:xmap}
    \iftoggle{conference}{\textstyle}{}
    \xbfhat = \argmin_{\xbf} \left[  f_1(\xbf)+f_2(\xbf) \right].
\eeq
\item \textbf{MMSE estimation and approximate inference}:
\iftoggle{conference}{Here}{In \emph{minimum mean-squared error} (MMSE) estimation,}
we wish to compute the posterior mean $\Exp(\xbf|\ybf)$ and maybe
also approximations of the
posterior covariance $\Cov(\xbf|\ybf)$ or
marginal posterior densities $\{p(x_n|\ybf)\}_{n=1}^N$.
\end{itemize}

For the MAP estimation problem \eqref{eq:xmap},
the separable structure of the objective function can be exploited by
one of several optimization methods, including variants of the
\emph{iterative shrinkage and thresholding algorithm} (ISTA)
\cite{ChamDLL:98}%
\iftoggle{conference}{}{
,\cite{DaubechiesDM:04,WrightNF:09,BeckTeb:09,Nesterov:07,BioDFig:07}}
and the \emph{alternating direction method of multipliers} (ADMM)
\cite{BoydPCPE:09}%
\iftoggle{conference}{.}{
,\cite{Esser:JIS:10,Chambolle:JMIV:11,He:JIS:12}.}

The MMSE and inference problems, however, are more difficult%
\iftoggle{conference}{}{\cite{pereyra2016stoch}},
even for the case of convex penalties
\cite{seeger2011fast,rangan2015admm}.
In recent years, there has been considerable interest
in approximations of \emph{loopy belief propagation}%
\iftoggle{conference}{}{
\cite{Pearl:88,YedidiaFW:03}}
for both MMSE estimation and approximate inference.  These methods
include variants of \emph{expectation propagation} (EP)
\cite{Minka:01,opper2004expectation,seeger2005expectation}
and, more recently, \emph{approximate message passing} (AMP)
\cite{DonohoMM:09,Rangan:11-ISIT}%
\iftoggle{conference}{}{
,\cite{DonohoMM:10-ITW1,DonohoMM:10-ITW2,Rangan:10-CISS}}.
For a posterior of the form \eqref{eq:pxy}, both EP and AMP
reduce the inference problem
to a sequence of problems involving only one penalty at a time.
These ``local" problems are computationally tractable under suitable penalties.
Moreover, in certain large random instances, these methods
are provably optimal \cite{BayatiM:11}%
\iftoggle{conference}{.}{,\cite{Rangan:11-ISIT,JavMon:12-arXiv}.}
Due to their generality,
these methods have been successfully applied to a wide range of problems,
e.g., \cite{fletcher2014scalable,FletcherRVB:11}%
\iftoggle{conference}{.}{%
,\cite{FletcherRVB:11,Schniter:11,SomS:12,parker2013bilinear2,Schniter:TSP:15,vila2015hyperspectral,Ziniel:TSP:15}.
}

Despite their computational simplicity, the convergence and accuracy of these
methods are not fully understood.
This work analyzes one promising EP-type method
known as \emph{expectation consistent approximate inference} (EC), originally proposed by
Opper and Winther in \cite{opper2004expectation}.
As shown in \cite{seeger2005expectation},
EC interpreted as a parallel form of the EP
method from \cite{Minka:01}, while being closely related to
the \emph{adaptive TAP} method from%
\iftoggle{conference}{}{\cite{opper2000gaussian}}
\cite{opper2001adaptive}.

As we now describe, our work contributes to the extension and understanding of Opper and Winther's EC method.
\begin{itemize}
\item \textbf{Generalization:}  We propose and analyze
a generalization of the EC
algorithm that we call \emph{Generalized EC} (GEC).  The proposed method can be applied to
arbitrary penalties $f_1(\xbf)$ and $f_2(\xbf)$,
and can also be used for both MAP or MMSE inference%
\iftoggle{conference}{.}{~by appropriate selection of estimation functions.}
Standard EC typically applies only to MMSE inference, often with one penalty
being quadratic.
Also, GEC supports a generalization of the covariance diagonalization step,
which is one of the key computational bottlenecks in standard EC
\cite{seeger2011fast}.

\item \textbf{Fixed points:}  It is well known that, when the standard
EC algorithm converges, its fixed points can be interpreted as
saddle points of an energy function~\cite{opper2004expectation,seeger2005expectation}
similar to the \emph{Bethe Free Energy} (BFE) that arises in the analysis of loopy BP
\cite{YedidiaFW:03}.  We provide a similar
energy-function interpretation of the MMSE-GEC algorithm (Theorem~\ref{thm:mmseFix}).
\iftoggle{conference}{}{
Our analysis shows that the so-called first- and second-order
terms output by MMSE-GEC can be interpreted as estimates
 of the posterior mean and variance.
Regarding the fixed points of MAP-GEC, we show that the first-order terms are critical
points of the objective function~\eqref{eq:xmap} and the second-order terms can be interpreted
as estimates of the local curvature of the objective function.
}

\item \textbf{Convergence:}  A critical concern for both EP and AMP is convergence
\cite{seeger2011fast,RanSchFle:14-ISIT}%
\iftoggle{conference}{.}{
\cite{Caltagirone:14-ISIT,Krzakala:14-ISITbethe}.
This situation is perhaps not surprising, given that they derive from loopy BP, which also
may diverge.
Most provably convergent alternate approaches are based on variants
of double-loop methods
such as \cite{opper2004expectation,rangan2015admm}.
Other modifications to improve the stability include
damping and fractional updates \cite{Seeger:08,RanSchFle:14-ISIT,Vila:ICASSP:15}
 and sequential updating \cite{manoel2014swamp}, which increase the
likelihood of convergence at the cost of convergence speed.
}
Our analysis of GEC convergence considers the first- and second-order
terms separately---a decoupling technique also used in \cite{seeger2005expectation}%
\iftoggle{conference}{}{,\cite{nickisch2009convex}}.
We show that, for strictly convex, smooth penalties, the standard updates for the first-order terms
are provably convergent.  For MAP-GEC, the second-order terms converge as well.

\item \textbf{Relation to the replica prediction of optimality:}
In \cite{tulino2013support}, Tulino et al.\ used a replica analysis from statistical physics
to predict the the MMSE error
when estimating a random vector $\xbf$ from noisy measurements of the linear transformation $\Abf\xbf$ under large, unitarily invariant, random $\Abf$.
\iftoggle{conference}{}{
This work extended the replica analyses in \cite{TanakaO:05,GuoV:05,RanganFG:12-IT}, which applied
to i.i.d.\ $\Abf$.  (See also \cite{vehkapera2014analysis}.)}
In \cite{cakmak2014samp,cakmak2015samp}, \c{C}akmak et al.\
proposed a variant of AMP (called S-AMP) using closely related methods.
In this work, we show that, when GEC is applied to linear regression, a prediction of the posterior MSE satisfies a fixed point equation that exactly matches
the replica prediction from \cite{tulino2013support}.

\iftoggle{conference}{}{
\item \textbf{Relation to ADMM:}
ADMM \cite{BoydPCPE:09} is a
popular approach to optimization problems of the form \eqref{eq:xmap} with
convex $f_i$.
Essentially, ADMM iterates individual optimizations of $f_1$ and $f_2$
together with a ``dual update'' that (asymptotically) enforces consistency
between the individual optimizations.
The dual update involves a fixed step size,
whose choice affects the convergence speed of the algorithm.
In this work, we show that GEC can be interpreted as a variant of ADMM with
two dual-updates, each with a step size that is \emph{adapted} according
to the local curvature of the corresponding penalty $f_i$.
}

\end{itemize}

\section{The Generalized EC Algorithm}

\subsection{Estimation and Diagonalization}

The proposed GEC algorithm involves two key operations:
i) \emph{estimation}, which computes an estimate of $\xbf$
using one penalty at a time; and ii) \emph{diagonalization} of
a sensitivity term.

\paragraph*{\bf Estimation}
The estimation function is constructed differently for
the MAP and MMSE cases.
In the MAP case, the estimation function is given by
\beq \label{eq:gimap}
    \iftoggle{conference}{\textstyle}{}
    \gbf_i(\rbf_i,\gammabf_i) := \argmin_{\xbf}
    \left[ f_i(\xbf)
    + \frac{1}{2}\|\xbf - \rbf_i\|^2_{\gammabf_i} \right] ,
\eeq
where $\rbf_i,\gammabf_i\in\R^N$ and $\gammabf_i > 0$ (componentwise),
and where
\[
    \iftoggle{conference}{\textstyle}{}
    \|\vbf\|^2_{\gammabf} := \sum_{n=1}^N \gamma_n|v_n|^2
\]
for any $\vbf$ and positive $\gammabf$.
\iftoggle{conference}{}{
The estimation function \eqref{eq:gimap} is a scaled version of
what is often called the \emph{proximal operator}.
}

For the MMSE problem, the estimation function is
\beq \label{eq:gimmse}
    \gbf_i(\rbf_i,\gammabf_i) := \Exp\left[ \xbf | \rbf_i,\gammabf_i \right],
\eeq
where the expectation is with respect to the conditional density
\beq \label{eq:pirgam}
    \iftoggle{conference}{\textstyle}{}
    p_i(\xbf|\rbf_i,\gammabf_i) = \frac{1}{Z_i} \exp\left[
         -f_i(\xbf)
    - \frac{1}{2}\|\xbf - \rbf_i\|^2_{\gammabf_i} \right].
\eeq

\paragraph*{\bf Diagonalization}
In its more general form, the diagonalization operator $\dbf(\Qbf)$ is an affine linear map from $\Qbf\in\R^{N\times N}$ to $\R^N$.
Several instances of diagonalization are relevant to our work.
For example, \emph{vector-valued diagonalization},
\beq \label{eq:Dvec}
    \dbf(\Qbf) := \diag(\Qbf),
\eeq
which simply returns a $N$-dimensional vector containing the diagonal elements
of $\Qbf$, and \emph{uniform diagonalization},
\beq \label{eq:Dunif}
    \dbf(\Qbf) := N^{-1} \Tr(\Qbf)\mathbf{1}_N,
\eeq
which returns a constant vector containing the \emph{average} diagonal element of $\Qbf$.
Here, $\mathbf{1}_N$ denotes the $N$-dimensional vector with all elements equal to one.

\iftoggle{conference}{}{
For the separable GLM, it will be useful to consider a block uniform
diagonalization.  In this case, we partition 
\beq \label{eq:xpart}
    \xbf = (\xbf_1;\cdots;\xbf_L), \quad \xbf_\ell \in \R^{N_\ell},
\eeq
with $\sum_{\ell} N_{\ell} = N$.  Conformal to the partition, we define the
\emph{block uniform diagonalization}
\beq \label{eq:Dblk}
    \dbf(\Qbf) := (d_1\mathbf{1}_{N_1};\cdots; d_L\mathbf{1}_{N_L}), \quad
    d_\ell = \frac{1}{N_{\ell}} \Tr(\Qbf_{\ell\ell}),
\eeq
where $\Qbf_{\ell\ell}\in\R^{N_\ell \times N_\ell}$ is the $\ell$-th diagonal block of $\Qbf$.

We note that any of these diagonalization operators
can be used with either the MAP or MMSE estimation functions.
}

\subsection{Algorithm Description}

The \emph{generalized EC} (GEC) algorithm is specified in Algorithm~\ref{algo:GEC}.
There,
$\partial \gbf_i(\rbf_i,\gammabf_i)/\partial \rbf_i$ is the $N\x N$ Jacobian matrix of $\gbf_i$ evaluated at $(\rbf_i,\gammabf_i)$,
$\Diag(\gammabf_i)$ is the diagonal matrix whose diagonal equals $\gammabf_i$,
``$\vdiv$'' is componentwise vector division, and
``$\vmult$'' is componentwise vector multiplication.
Note that it is not necessary to compute the full matrix $\Qbf_i$
in line~\ref{line:Qi}; it suffices to compute only the diagonalization $\dbf(\Qbf_i)$.

\begin{algorithm}
\caption{Generalized EC (GEC)}
\begin{algorithmic}[1]  \label{algo:GEC}
\REQUIRE{Estimation functions $\gbf_1(\cdot,\cdot)$, $\gbf_2(\cdot,\cdot)$
and diagonalization operator $\dbf(\cdot)$.}
\STATE{ Select initial $\rbf_1,\gammabf_1$ }
\REPEAT
    \FOR{ $(i,j) = (1,2)$ and $(2,1)$}
        \STATE{$\xbfhat_i \gets \gbf_i(\rbf_i,\gammabf_i)$} \label{line:xi}
        \STATE{$\Qbf_i \gets
         [\partial \gbf_i(\rbf_i,\gammabf_i)/\partial \rbf_i] \Gammabf_i^{-1},~
         \Gammabf_i = \Diag(\gammabf_i)$} \label{line:Qi}
        \STATE{$\etabf_i \gets \mathbf{1}\vdiv\dbf( \Qbf_i )$}
            \label{line:etai}
        \STATE{$\gammabf_j \gets \etabf_i - \gammabf_i$}  \label{line:gamj}
        \STATE{$\rbf_j \gets (\etabf_i \vmult \xbfhat_i - \gammabf_i \vmult \rbf_i)\vdiv\gammabf_j$}
            \label{line:rj}
    \ENDFOR
\UNTIL{Terminated}
\end{algorithmic}
\end{algorithm}

It will sometimes be useful to
rewrite Algorithm~\ref{algo:GEC} in a \emph{scaled form}.
Define $\betabf_i := \gammabf_i \vmult \rbf_i$ and $\gt_i(\betabf_i,\gammabf_i)
:= \gbf_i(\betabf_i \vdiv \gammabf_i,\gammabf_i)$.
Then GEC can be rewritten as
\begin{subequations} \label{eq:GECs}
\begin{align}
    \etabf_i &\gets \mathbf{1}\vdiv\dbf( \Qbftilde_i ), ~
    \Qbftilde_i :=
    \partial \gt_i(\betabf_i,\gammabf_i)/\partial \betabf_i \label{eq:etais} \\
    \gammabf_j &\gets \etabf_i - \gammabf_i \label{eq:nugam1} \\
    \betabf_j &\gets \etabf_i \vmult \gt_i(\betabf_i,\gammabf_i) - \betabf_i. \label{eq:betajs}
\end{align}
\end{subequations}

Note that, in line \ref{line:Qi} of Algorithm~\ref{algo:GEC}, we are required to compute the (scaled) Jacobian of the
estimation function.
For the MAP estimation function~\eqref{eq:gimap}, this quantity becomes \cite{Rangan:11-ISIT}
\beq \label{eq:gimapDeriv}
    [\partial \gbf_i(\rbf_i,\gammabf_i)/\partial \rbf_i]\Gammabf_i^{-1} =
    \left[ \Hessian{f_i}(\xbfhat_i) + \Gammabf_i \right]^{-1},
\eeq
where $\xbfhat_i$ is the minimizer in \eqref{eq:gimap}
and $\Hessian{f_i}(\xbfhat_i)$ is the Hessian of $f_i$ at that
minimizer.    For the MMSE estimation function,
this scaled Jacobian becomes the covariance matrix
\beq \label{eq:gimmseDeriv}
    [\partial \gbf_i(\rbf_i,\gammabf_i)/\partial \rbf_i]\Gammabf_i^{-1} =
    \Cov(\xbf_i|\rbf_i,\gammabf_i),
\eeq
where the covariance is taken with respect to the density
\eqref{eq:pirgam}.

\subsection{Examples} \label{sec:ex}

\paragraph*{\bf SLR with Separable Prior}
Suppose that we aim to estimate $\xbf$ given noisy linear measurements of the form
\beq \label{eq:yawgn}
    \ybf = \Abf \xbf + \wbf, \quad
    \wbf \sim {\mathcal N}(0,\gamma_w^{-1} \Ibf),
\eeq
where $\Abf\in\R^{M\times N}$ is a known matrix and
$\wbf$ is independent of $\xbf$.
Statisticians often refer to this problem as \emph{standard linear regression} (SLR).
Suppose also that $\xbf$ has independent elements with marginal densities $p(x_n)$:
\beq \label{eq:pxindep}
    \iftoggle{conference}{\textstyle}{}
    p(\xbf) = \prod_{n=1}^N p(x_n) .
\eeq
Then, the posterior $p(\xbf|\ybf)$ takes the form of \eqref{eq:pxy} when
\beq \label{eq:fawgn}
    \iftoggle{conference}{\textstyle}{}
    f_1(\xbf) = -\sum_{n=1}^N \log p(x_n), \quad
    f_2(\xbf) = \frac{\gamma_w}{2} \| \ybf - \Abf\xbf\|^2.
\eeq
The separable nature of $f_1(\xbf)$ implies that, in both the MAP or MMSE cases,
the output of the estimator $\gbf_1$
(recall \eqref{eq:gimap} and \eqref{eq:gimmse}) can be computed in
a componentwise manner, as can the diagonal terms of their Jacobians.
Likewise, the quadratic nature of $f_2(\xbf)$ implies that
the output of $\gbf_2$ can be computed by solving a linear system.

\paragraph*{\bf GLM with Separable Prior}  Now suppose that,
instead of \eqref{eq:yawgn}, we have a more general likelihood with the form
\beq \label{eq:GLM}
    \iftoggle{conference}{\textstyle}{}
    p(\ybf|\xbf) = \prod_{m=1}^M p(y_m|z_m), \quad \zbf = \Abf\xbf .
\eeq
Statisticians often refer to \eqref{eq:GLM} as the \emph{generalized
linear model} (GLM)%
\iftoggle{conference}{.}{ \cite{NelWed:72,McCulNel:89}. }
To pose the GLM in a format convenient for GEC,
we define the new vector $\ubf = (\xbf;\zbf)$.  Then, the posterior
$p(\ubf|\ybf)=p(\xbf,\zbf|\ybf)$ can be placed in the form of \eqref{eq:pxy} using the penalties\!
\begin{subequations} \label{eq:fglm}
\begin{align*}
    f_1(\ubf)
    &= f_1(\xbf,\zbf)
    \iftoggle{conference}{\textstyle}{}
    = -\sum_{n=1}^N \log p(x_n) - \sum_{m=1}^M \log p(y_m|z_m), \\
    f_2(\ubf)
    &= f_2(\xbf,\zbf) = \begin{cases}
        0       & \mbox{if } \zbf=\Abf\xbf, \\
        \infty & \mbox{if } \zbf \neq \Abf\xbf
    \iftoggle{conference}{.}{,}
        \end{cases}
\end{align*}
\end{subequations}
\iftoggle{conference}{}{where $f_2(\ubf)$ constrains $\ubf$ to the nullspace of $[\Abf~-\Ibf]$.}
Because the first penalty $f_1(\ubf)$ remains separable,
the MAP and MMSE functions can be evaluated componentwise, as in separable SLR.
For the second penalty $f_2(\ubf)$, MAP or MMSE estimation simply becomes projection onto a linear space.

\section{Fixed Points of GEC}

\subsection{Consistency}
We now characterize the fixed points of GEC
for both MAP and MMSE estimation functions.
For both scenarios, we will need the following simple consistency
result.

\begin{lemma} \label{lem:fixCon}
Consider GEC (Algorithm~\ref{algo:GEC}) with arbitrary estimation
functions $\gbf_i(\cdot,\cdot)$ and arbitrary diagonalization operator
$\dbf(\cdot)$.  For any fixed point with
$\gammabf_1+\gammabf_2 > 0$, we have
\begin{subequations}
\begin{align}
    \etabf_1 &= \etabf_2 = \etabf := \gammabf_1 + \gammabf_2
        \label{eq:etafix} \\
    \xbfhat_1 &= \xbfhat_2 = \xbfhat :=
        \left( \gammabf_1\rbf_1 + \gammabf_2\rbf_2 \right)
        \vdiv (\gammabf_1+\gammabf_2)
        \label{eq:xhatfix} .
\end{align}
\end{subequations}
\end{lemma}
\begin{proof}  From line~\ref{line:gamj} of Algorithm~\ref{algo:GEC},
$\etabf_i = \gammabf_1 + \gammabf_2$
for $i=1,2$, which proves \eqref{eq:etafix}.
Also, since $\gammabf_1+\gammabf_2>0$, the elements of $\etabf$
are invertible.
In addition, from line~\ref{line:rj},
\[
    \xbfhat_i = \left( \gammabf_1\vmult\rbf_1 + \gammabf_2\vmult\rbf_2 \right)
    \vdiv \etabf_i \text{~~for~~} i=1,2,
\]
which proves \eqref{eq:xhatfix}.
\end{proof}

\subsection{MAP Estimation} \label{sec:mapFix}

We first examine GEC's fixed points for MAP estimation.

\begin{theorem} \label{thm:mapFix}
Consider GEC (Algorithm~\ref{algo:GEC})
with the MAP estimation functions from~\eqref{eq:gimap}
and an arbitrary diagonalization $\dbf(\cdot)$.
For any fixed point with $\gammabf_i > 0$,
let $\xbfhat=\xbfhat_i$ be the common value
of the two estimates as defined in Lemma~\ref{lem:fixCon}.
Then $\xbfhat$ is a stationary point of the minimization~\eqref{eq:xmap}.
\end{theorem}
\begin{proof}
\iftoggle{conference}{
See the full paper \cite{Fletcher:arxiv:16}.
}{
See Appendix~\ref{sec:mapFixPf}.
}
\end{proof}

\subsection{MAP Estimation and Curvature} \label{sec:curvature}

Note that Theorem~\ref{thm:mapFix} applies
to an arbitrary diagonalization operator $\dbf(\cdot)$.
This raises two questions:
i) what is the role of the diagonalization operator $\dbf(\cdot)$, and
ii) how can the fixed point $\etabf$ be interpreted as a result of that diagonalization?
\iftoggle{conference}{}{
We now show that, under certain additional conditions and certain choices
of $\dbf$, $\etabf$ can be related to the \emph{curvature}
of the optimization objective in \eqref{eq:xmap}. }

Let $\xbfhat$
be a stationary point of \eqref{eq:xmap} and let $\Pbf_i = \Hessian{f_i}(\xbfhat)$ be the Hessian
of $f_i$ at $\xbfhat$.  Then, the Hessian of the objective function in
\eqref{eq:xmap} is $\Pbf_1+\Pbf_2$.  Furthermore, let
\beq \label{eq:etahat}
    \etabfhat := \mathbf{1}\vdiv\dbf\left( (\Pbf_1+\Pbf_2)^{-1} \right),
\eeq
so that $\mathbf{1}\vdiv\etabfhat$ is the diagonal of the inverse Hessian.  Geometrically speaking,
this inverse Hessian measures the curvature of the objective function at the critical point $\xbfhat$.

We now identify two cases where $\etabf = \etabfhat$:
i) when $\Pbf_i$ are diagonal, and ii) when $\Pbf_i$ are \emph{free}.
To define ``free,'' consider the Stieltjes transform $S_{\Pbf}(\omega)$
of any real symmetric matrix $\Pbf$:
\beq \label{eq:stieltjes}
    \iftoggle{conference}{\textstyle}{}
    S_{\Pbf}(\omega) = \frac{1}{N} \Tr\left[ (\Pbf - \omega \Ibf)^{-1} \right]
     = \frac{1}{N} \sum_{n=1}^N \frac{1}{\lambda_n - \omega},
\eeq
where $\lambda_n$ are the eigenvalues of $\Pbf$. Also, let $R_{\Pbf}(\omega)$
denote the so-called $R$-transform of $\Pbf$, given by
\beq \label{eq:rtrans}
    \iftoggle{conference}{\textstyle}{}
    R_{\Pbf}(\omega) =  S_{\Pbf}^{-1}(-\omega) - \frac{1}{\omega},
\eeq
where the inverse $S_{\Pbf}^{-1}(\cdot)$ is in terms of composition of functions.
The Stieltjes and $R$-transforms are discussed in detail in \cite{TulinoV:04}.
We will say that $\Pbf_1$ and $\Pbf_2$ are ``free'' if
\beq \label{eq:free}
    R_{\Pbf_1 + \Pbf_2}(\omega) = R_{\Pbf_1}(\omega) + R_{\Pbf_2}(\omega).
\eeq

An important example of freeness is the following.
Suppose that the penalty functions are given by
$f_i(\xbf) = h_i(\Abf_i\xbf)$ for some matrices $\Abf_i$ and functions $h_i(\cdot)$.
Then
\[
    \Pbf_i = \Hessian{f_i}(\xbfhat) =  \Abf_i\tran \Hessian{h_i}(\zbfhat_i) \Abf_i, \quad \zbfhat_i=\Abf_i\xbfhat.
\]
It is shown in \cite{TulinoV:04} that, if $\zbfhat_i$ are fixed and $\Abf_i$
are unitarily invariant random matrices, then $\Pbf_i$ are asymptotically
free in certain limits as $N \arr \infty$.  Freeness will thus occur in the limits
of large problem with unitarily invariant random matrices.

\begin{theorem} \label{thm:curve}
Consider GEC (Algorithm~\ref{algo:GEC}) with the MAP estimation functions~\eqref{eq:gimap}.
Consider any fixed point with $\gammabf_i>0$, and let $\xbfhat$ and $\etabf$ be the
common values of $\xbfhat_i$ and $\etabf_i$ from Lemma~\ref{lem:fixCon}. Recall
that $\xbfhat$ is a stationary point of the minimization \eqref{eq:xmap}
via Theorem~\ref{thm:mapFix}.
Then $\etabf = \etabfhat$ from \eqref{eq:etahat} under either
\begin{enumerate}[(a)]
\item vector-valued $\dbf(\cdot)$ from \eqref{eq:Dvec} and diagonal $\Pbf_i$; or
\item uniform $\dbf(\cdot)$ from \eqref{eq:Dunif} and free $\Pbf_i$.
\end{enumerate}
\end{theorem}
\begin{proof}
\iftoggle{conference}{
See the full paper \cite{Fletcher:arxiv:16}.
}{See Appendix~\ref{sec:curvePf}.
}
\end{proof}

\subsection{MMSE Estimation} \label{sec:mmseFix}

We now consider the fixed points of GEC under MMSE estimation functions.
It is well-known that
the fixed points of the standard EC algorithm are critical points
of a certain free-energy optimization for approximate inference
\cite{opper2004expectation,seeger2005expectation}.
We derive a similar characterization of the fixed points of GEC.

Let $p(\xbf|\ybf)$ be the density \eqref{eq:pxy} for some fixed $\ybf$.
Then, given any density $b(\xbf)$, it is straightforward to show that
the KL divergence between $b(\xbf)$ and $p(\xbf|\ybf)$ can be expressed as
\beq \label{eq:FE}
    D(b\|p) = D(b\|e^{-f_1}) + D(b\|e^{-f_2}) + H(b) + \mbox{const},
\eeq
where $H(b)$ is the differential entropy of $b$ and
the constant term does not depend on $b$.  Thus, in principle,
we could compute $p$ by minimizing \eqref{eq:FE} over all densities $b$.
Of course, this minimization is generally intractable since it involves
a search over an $N$-dimensional density.

To approximate the minimization, define
\beq \label{eq:JBFE}
    J(b_1,b_2,q) := D(b_1\|e^{-f_1}) + D(b_2\|e^{-f_2}) + H(q),
\eeq
where $b_1$, $b_2$ and $q$ are densities on the variable $\xbf$.
Note that minimization of \eqref{eq:FE}
over $b$ is equivalent to the optimization
\iftoggle{conference}{
\begin{align}
    & \textstyle \min_{b_1,b_2} \max_q J(b_1,b_2,q) \label{eq:JBFEOpt} \\
    & \text{\,such that~~} b_1 = b_2 = q.  \label{eq:BFEeqcon}
\end{align}
}{
\beq \label{eq:JBFEOpt}
    \min_{b_1,b_2} \max_q J(b_1,b_2,q)
\eeq
under the constraint
\beq \label{eq:BFEeqcon}
    b_1 = b_2 = q.
\eeq
}
The energy function \eqref{eq:JBFE} is known as the \emph{Bethe Free Energy}
(BFE).  Under the constraint \eqref{eq:BFEeqcon},
the BFE matches the original energy function~\eqref{eq:FE}.
However, BFE minimization under the constraint~\eqref{eq:BFEeqcon} is
equally intractable.

As with EC, the GEC
algorithm can be derived as a relaxation of the above BFE optimization,
wherein \eqref{eq:BFEeqcon} is replaced by the so-called
\emph{moment matching} constraints:
\begin{subequations} \label{eq:MMcon}
\begin{align}
    \Exp(\xbf|b_1) &= \Exp(\xbf|b_2) = \Exp(\xbf|q)
        \label{eq:MMconmean} \\
    \dbf(\Exp(\xbf\xbf\tran|b_1)) &= \dbf(\Exp(\xbf\xbf\tran|b_2))
    = \dbf(\Exp(\xbf\xbf\tran|q)). \label{eq:MMconvar}
\end{align}
\end{subequations}
Thus, instead of requiring a perfect match in the densities $b_1,b_2,q$ as in \eqref{eq:BFEeqcon},
GEC requires only a match in
their first moments and certain diagonal components of their second moments.
Note that, for the vector-valued diagonalization
\eqref{eq:Dvec}, \eqref{eq:MMconvar} is equivalent to
\iftoggle{conference}{%
    $\Exp\left[ x_n^2 \,|\, b_i \right] = \Exp\left[ x_n^2 \,|\, q \right]
    ~\forall n, i$,
}{%
\[
    \Exp\left[ x_n^2 \,|\, b_i \right] = \Exp\left[ x_n^2 \,|\, q \right]
    ~~\forall n, i ,
\]}%
which requires only that the marginal 2nd moments match.
Under the uniform diagonalization \eqref{eq:Dunif},
\eqref{eq:MMconvar} is equivalent to
\[
    \iftoggle{conference}{\textstyle}{}
    \frac{1}{N} \sum_{n=1}^N
    \Exp\left[ x_n^2 \,|\, b_i \right] =
    \frac{1}{N} \sum_{n=1}^N \Exp\left[ x_n^2 \,|\, q \right],
    ~~i=1,2 ,
\]
requiring only that the \emph{average} 2nd marginal moments match.

\begin{theorem} \label{thm:mmseFix}
Consider GEC (Algorithm~\ref{algo:GEC}) with the
MMSE estimation functions \eqref{eq:gimmse} and either vector-valued
\eqref{eq:Dvec} or uniform \eqref{eq:Dunif} diagonalization.
For any fixed point with $\gammabf_i>0$,
let $\xbfhat$ and $\etabf$ be the common values of $\xbfhat_i$ and $\etabf_i$
from Lemma~\ref{lem:fixCon}. Also let
\beq \label{eq:bifix}
    b_i(\xbf) = p_i(\xbf|\rbf_i,\gammabf_i)
\eeq
for $p_i(\xbf|\rbf_i,\gammabf_i)$ from \eqref{eq:pirgam} and
let $q(\xbf)$ be the Gaussian density
\beq \label{eq:qfix}
    \iftoggle{conference}{\textstyle}{}
    q(\xbf) \propto \exp\left[ -\frac{1}{2} \|\xbf-\xbfhat\|^2_{\etabf} \right].
\eeq
Then, $b_1,b_2,q$ are stationary points of the optimization
\eqref{eq:JBFEOpt} subject to the moment matching constraints \eqref{eq:MMcon}.
In addition, $\xbfhat$ is the mean, and $\etabf$
the marginal precision, of these densities:
\begin{align}
    \xbfhat      &= \Exp(\xbf|q) = \Exp(\xbf|b_i), ~i=1,2 \\
    \mathbf{1}\vdiv\etabf  &= \dbf(\Cov(\xbf\xbf\tran|q)) = \dbf(\Cov(\xbf\xbf\tran|b_i)), ~i=1,2.
\end{align}
\end{theorem}
\begin{proof}
\iftoggle{conference}{
See the full paper \cite{Fletcher:arxiv:16}.
}{See Appendix~\ref{sec:mmseFixPf}.
}
\end{proof}

\subsection{An Unbiased Estimate of $\xbf$} \label{sec:suffstat}

As described in Section~\ref{sec:ex}, a popular application of GEC is
to approximate the marginals of the posterior density \eqref{eq:pxy}
in the case that the first penalty $f_1(\xbf)$ describes
the prior and the second penalty
$f_2(\xbf;\ybf)$ describes the likelihood.
That is,
\[
    p(\xbf) \propto e^{-f_1(\xbf)} \text{~~and~~}
    p(\ybf|\xbf) \propto e^{-f_2(\xbf; \ybf)}.
\]
\iftoggle{conference}{}{Here, we have made the dependence of $f_2(\xbf;\ybf)$ on $\ybf$ explicit. }
The GEC algorithm produces three estimates for the posterior density:
$b_1$, $b_2$, and $q$.  Consider the first of these estimates, $b_1$.
From~\eqref{eq:bifix} and \eqref{eq:pirgam}, this belief estimate is given by
\beq \label{eq:b1xr}
    \iftoggle{conference}{\textstyle}{}
    b_1(\xbf;\rbf_1,\gammabf_1)
    = Z(\rbf_1)^{-1} p(\xbf)
        \exp\left[-\frac{1}{2}\|\xbf-\rbf_1\|^2_{\gammabf_1}\right],
\eeq
where $Z(\rbf_1)$ is a normalization constant.

\iftoggle{conference}{
In the full paper \cite{Fletcher:arxiv:16}, we argue that $\rbf_1$ can be
interpreted as a random variable that, under a natural choice of prior, yields
}{
If we model $\rbf_1$ as a random vector, then \eqref{eq:b1xr} implies that
\[
    p(\xbf|\rbf_1) = b_1(\xbf;\rbf_1,\gammabf_1) ,
\]
From Bayes rule, we know $p(\xbf|\rbf_1) = p(\xbf)p(\rbf_1|\xbf)/p(\rbf_1)$,
which together with \eqref{eq:b1xr} implies
\[
    p(\rbf_1|\xbf)
    = \frac{p(\rbf_1)}{Z(\rbf_1)}
        \exp\left[-\frac{1}{2}\|\rbf_1-\xbf\|^2_{\gammabf_1}\right] .
\]
For $p(\rbf_1)$ to be an admissible prior density on $\rbf_1$, it must satisfy $p(\rbf_1)\geq 0$, $\int p(\rbf_1) \dif\rbf_1=1$, and $\int p(\rbf_1|\xbf)\dif\rbf_1=1~\forall \xbf$.
It is straightforward to show that one admissible choice is
\[
    p(\rbf_1) = cZ(\rbf_1), \quad c^2 = (2\pi)^{-N}
    \prod_{n=1}^N \gamma_{1n}.
\]
Under this choice, we get}
\beq \label{eq:prx}
    p(\rbf_1|\xbf)
    ={\mathcal N}(\xbf,\Gammabf_1^{-1}) ,
\eeq%
in which case $\rbf_1$ can be interpreted as an unbiased estimate of $\xbf$ with $\Gammabf_1^{-1}$-covariance Gaussian estimation error.

The situation above is reminiscent of AMP algorithms \cite{DonohoMM:09,Rangan:11-ISIT}.
Specifically, their state evolution analyses \cite{BayatiM:11}
show that, under large i.i.d.\ $\Abf$, they
recursively produce a sequence of vectors $\{\rbf^k\}_{k\geq 0}$
that can be modeled as realizations of the true vector $\xbf$ plus
zero-mean white Gaussian noise.
\iftoggle{conference}{}{
They then compute a sequence of estimates of
$\xbf$ by ``denoising" each $\rbf^k$.}

\section{Convergence of the First-Order Terms for Strictly Convex Penalties}

We first analyze the convergence of GEC with fixed
``second-order terms'' $\etabf_i$ and $\gammabf_i$.
To this end, fix $\gammabf_i>0$ at arbitrary values and assume that $\etabf_i$ are fixed
points of \eqref{eq:nugam1}.
Then Lemma~\ref{lem:fixCon} implies that $\etabf_1=\etabf_2=\etabf:=\gammabf_1+\gammabf_2$.
With $\etabf_i$ and $\gammabf_i$ fixed, the (scaled) GEC algorithm
\eqref{eq:GECs} updates only $\betabf_i=\gammabf_i\vmult\xbf_i$.
In particular, \eqref{eq:betajs} implies that this update is
\beq \label{eq:betaFO}
    \betabf_j \gets (\Gammabf_1+\Gammabf_2)\gt_i(\betabf_i,\gammabf_i) - \betabf_i,
    ~(i,j)\in\{(1,2),(2,1)\} .
\eeq
\iftoggle{conference}{
We analyze the recursion \eqref{eq:betaFO} under the following assumption
-- the full paper \cite{Fletcher:arxiv:16} provides a slightly more general
condition.

\begin{assumption} \label{as:firstOrder}
Suppose that $f_i(\xbf)$ is strictly convex and smooth in that
its Hessian satisfies
\beq \label{eq:fihessbnd}
    c_{i1}\Ibf \leq \Hessian{f_i}(\xbf) \leq c_{i2}\Ibf,
\eeq
for constants $c_{i1}, c_{i1} > 0$ and all $\xbf$.
\end{assumption}
}{
We analyze the recursion \eqref{eq:betaFO} under the following assumption

\begin{assumption} \label{as:firstOrder}
For $i=1,2$, fix $\gammabf_i > 0$ and suppose that
$\gt_i(\betabf_i,\gammabf_i)$ is differentiable in $\betabf_i$.
Also define
\beq \label{eq:Qderas}
    \Qbftilde_i(\betabf_i) :=
    \partial \gt_i(\betabf_i,\gammabf_i)/\partial \betabf_i ,
\eeq
and assume that $\Qbftilde_i(\betabf_i)$ is symmetric and that
there exists constants $c_{i1}, c_{i2} > 0$ such that, for all $\betabf_i$,
\beq \label{eq:Qderivbnd}
     c_{i1}\Ibf + \Gammabf_i \leq \Qbftilde_i(\betabf_i)^{-1} \leq  c_{i2}\Ibf + \Gammabf_i.
\eeq
\end{assumption}

This assumption is valid under strictly convex penalties:

\begin{lemma} \label{lem:convFn}  Suppose that $f_i(\xbf)$ is strictly convex in that
its Hessian satisfies
\beq \label{eq:fihessbnd}
    c_{i1}\Ibf \leq \Hessian{f_i}(\xbf) \leq c_{i2}\Ibf,
\eeq
for constants $c_{i1}, c_{i1} > 0$ and all $\xbf$.  Then, the MAP and MMSE estimation
functions \eqref{eq:gimap} and \eqref{eq:gimmse} satisfy Assumption~\ref{as:firstOrder}.
\end{lemma}
\begin{proof}
\iftoggle{conference}{
See the full paper \cite{Fletcher:arxiv:16}.
}{
See \cite{Rangan:arxiv:15}.
}
\end{proof}
}

We then have the following convergence result.

\begin{theorem} \label{thm:convFO}
Consider the recursion \eqref{eq:betaFO}
\iftoggle{conference}{with the MAP or MMSE estimation functions with penalties
satisfying Assumption~\ref{as:firstOrder}}{
with functions $\gt_i(\cdot,\cdot)$ satisfying
Assumption~\ref{as:firstOrder}}
and arbitrary fixed values of $\gammabf_i > 0$, for $i=1,2$.
Then, from any initialization of $\betabf_i$, \eqref{eq:betaFO} converges to a unique fixed point
that is invariant to the choice of $\gammabf_i$.
\end{theorem}
\begin{proof}
\iftoggle{conference}{
See the full paper \cite{Fletcher:arxiv:16}.
}{
See Appendix \ref{sec:convFOPf}.
}
\end{proof}

\section{Convergence of the Second-Order Terms for MAP Estimation}

\subsection{Convergence}
We now examine the convergence of the second-order terms $\etabf_i$ and $\gammabf_i$.
The convergence results that we present here apply only to the
case of MAP estimation~\eqref{eq:gimap} under strictly convex penalties
$f_i(\xbf)$ that satisfy the conditions in
\iftoggle{conference}{Assumption~\ref{as:firstOrder}}{Lemma~\ref{lem:convFn}}.
Furthermore, they assume that Algorithm~\ref{algo:GEC} is initialized using a pair $(\rbf_1,\gammabf_1)$ yielding $\gbf_1(\rbf_1,\gammabf_1)=\xbfhat$, where $\xbfhat$ is a local minimizer of \eqref{eq:xmap}.

\begin{theorem} \label{thm:convMAP}  Consider GEC (Algorithm~\ref{algo:GEC})
under the MAP estimation functions~\eqref{eq:gimap} with penalties $f_i(\xbf)$ that
are strictly convex functions satisfying
\iftoggle{conference}{Assumption~\ref{as:firstOrder}.}{
the assumptions in Lemma~\ref{lem:convFn}.}
Construct $\rbf_1^0$ and $\gammabf_1^0$ as follows:
\begin{quote}
Choose arbitrary $\gammabf_1^0,\gammabf_2^0>0$ and run Algorithm~\ref{algo:GEC}
under \emph{fixed} $\gammabf_i=\gammabf_i^0$
and \emph{fixed} $\etabf_i=\gammabf_1^0+\gammabf_2^0$ (for $i=1,2$)
until convergence (as guaranteed by Theorem~\ref{thm:convFO}).
Then record the final value of $\rbf_1$ as $\rbf_1^0$.
\end{quote}
Finally, run Algorithm~\ref{algo:GEC} from the initialization
$(\rbf_1,\gammabf_1)=(\rbf_1^0,\gammabf_1^0)$
without keeping $\gammabf_i$ and $\etabf_i$ fixed.
\begin{enumerate}[(a)]
\item For all subsequent iterations, we will have
$\xbfhat_i=\xbfhat$, where $\xbfhat$ is the unique global minimizer of \eqref{eq:xmap}.
\item If $\dbf(\gammabf)$ is either the vector-valued or uniform
diagonalization operator,
then the second-order terms $\gammabf_i$ will converge to unique fixed points.
\end{enumerate}
\end{theorem}
\begin{proof}
\iftoggle{conference}{
See the full paper \cite{Fletcher:arxiv:16}.
}{See Appendix~\ref{sec:convMAPPf}.
}
\end{proof}

\section{Relation to the Replica Prediction}

Consider the separable SLR problem described in Section~\ref{sec:ex}
for any matrix $\Abf$ and noise precision $\gamma_w>0$.
Consider GEC under the penalty functions \eqref{eq:fawgn},
MMSE estimation~\eqref{eq:gimmse}, and uniform diagonalization~\eqref{eq:Dunif}.
Thus, $\gammabf_i$ will have identical components of scalar value $\gamma_i$.

Suppose that $b_1(\xbf)$ is the belief estimate generated at a fixed point of GEC.
Since $p(\xbf)$ in \eqref{eq:pxindep} is separable, \eqref{eq:b1xr} implies
\[
    \iftoggle{conference}{\textstyle}{}
    b_1(\xbf;\rbf_1,\gamma_1)
    \propto \prod_{n=1}^N p(x_n)e^{-\gamma_1(x_n-r_{1n})^2/2} .
\]
In the sequel, let
$\Exp(\cdot|r_{1n},\gamma_1)$ and $\var(\cdot |r_{1n},\gamma_1)$
denote the mean and variance with respect to the marginal density
\beq \label{eq:b1n}
    b_1(x_n|r_{1n},\gamma_1) \propto p(x_n)e^{-\gamma_1(x_n-r_{1n})^2/2}.
\eeq
From \eqref{eq:MMconmean}, the GEC estimate $\xbfhat$
satisfies $\xhat_n = \Exp(x_n|r_{1n},\gamma_1)$%
, which is the posterior mean under the estimated density \eqref{eq:b1n}.
Also, from \eqref{eq:MMconvar} and the definition of the uniform diagonal operator
\eqref{eq:Dunif}, the components of $\etabf$ are identical and satisfy
\beq \label{eq:etafixr}
    \iftoggle{conference}{\textstyle}{}
    \eta^{-1} = \frac{1}{N} \Tr(\Cov(\xbf|\rbf_1,\gamma_1))
        = \frac{1}{N} \sum_{n=1}^N \var(x_n|r_{1n},\gamma_1),
\eeq
which is the average of the marginal posterior variances.
Equivalently, $\eta^{-1}$ is the average estimation MSE,
\[
    \iftoggle{conference}{\textstyle}{}
    \eta^{-1}
    = \frac{1}{N} \sum_{n=1}^N \Exp\left[ (x_n-\xhat_n)^2|r_{1n},\gamma_1\right].
\]

We will show that the value for $\eta$ can be
characterized in terms of the singular values of $\Abf$.
Let $\Ybf := \gamma_w\Abf\tran\Abf$, and let $S_{\Ybf}(\omega)$ denote its
Stieltjes Transform \eqref{eq:stieltjes} and $R_{\Ybf}(\omega)$ its $R$-transform \eqref{eq:rtrans}.
We then have the following.

\begin{theorem} \label{thm:fixr}  For the above problem, at any fixed point
of GEC (Algorithm~\ref{algo:GEC}), $\eta$ and $\gamma_1$ satisfy the fixed-point equations
\beq \label{eq:fixr}
    \iftoggle{conference}{\textstyle}{}
    \gamma_1 = R_{\Ybf}(-\eta^{-1}), \quad
    \eta^{-1} = \frac{1}{N} \sum_{n=1}^N \var(x_n|r_{1n},\gamma_1),
\eeq
where $\var( x_n | r_{1n},\gamma_1)$ is the posterior variance from
the density in \eqref{eq:b1n}.
\end{theorem}
\begin{proof}
\iftoggle{conference}{
See the full paper \cite{Fletcher:arxiv:16}.
}{
See Appendix~\ref{sec:fixrPf}.
}
\end{proof}

It is interesting to compare this result with that in \cite{tulino2013support},
which considers exactly this estimation problem in the limit of large $N$
with certain unitarily invariant random matrices $\Abf$ and i.i.d.\ $x_n$.
That work uses
a replica symmetric (RS) analysis to predict that the asymptotic
MSE $\eta^{-1}$ satisfies the fixed point equations
\beq \label{eq:fixrep}
    \gamma_1 = R_{\Ybf}(-\eta^{-1}), \quad
    \eta^{-1} =  \Exp\left[
        \var( x_n | r_{1n}, \gamma_1 ) \right],
\eeq
where the expectation is over
$r_{1n}=x_n + {\mathcal N}(0,\gamma_1^{-1})$.
This Gaussian distribution is exactly the predicted likelihood
$p(r_{1n}|x_n)$ in \eqref{eq:prx}.  Hence, if $x_n$ is i.i.d., and $r_{1n}$
follows the likelihood in \eqref{eq:prx}, then the MSE predicted from the
GEC estimated posterior must satisfy the same fixed point equation as the
minimum MSE predicted from the replica method in the limit as $N \arr \infty$.
In particular, if
this equation has a unique fixed point, then the GEC-predicted MSE
will match the minimum MSE as given by the replica method.

Of course, these arguments are not a rigorous proof of optimality.
The analysis relies on the GEC model $p(\xbf|\rbf_1)$ with a particular choice of prior on $\rbf_1$.
Also, the replica method itself is not rigorous.
Nevertheless, the arguments do provide some hope that GEC is optimal
in certain asymptotic and random regimes.

\iftoggle{conference}{}{

\section{Relation to ADMM} \label{sec:admm}

We conclude by relating GEC to the well-known
alternating direction method of multipliers (ADMM)
\cite{BoydPCPE:09,Esser:JIS:10,Chambolle:JMIV:11,He:JIS:12}.
Consider the MAP minimization problem \eqref{eq:xmap}.
To solve this via ADMM, we rewrite the minimization as a constrained
optimization
\beq \label{eq:xmapsplit}
    \min_{\xbf_1,\xbf_2} f_1(\xbf_1) + f_2(\xbf_2) \mbox{ s.t. }
        \xbf_1 = \xbf_2.
\eeq
The division of the variable $\xbf$ into two variables $\xbf_1$ and
$\xbf_2$ is often called \emph{variable splitting}.
Corresponding to the constrained optimization \eqref{eq:xmapsplit},
define the augmented Lagrangian,
\begin{align}
    L_\gamma(\xbf_1,\xbf_2,\sbf)
    &= f_1(\xbf_1) + f_2(\xbf_2) + \sbf\tran(\xbf_1 -\xbf_2)
    \nonumber\\&\quad
       + \frac{\gamma}{2}\|\xbf_1-\xbf_2\|^2,
        \label{eq:Ladmm}
\end{align}
where $\sbf$ is a dual vector and $\gamma > 0$ is an adjustable weight.  The ADMM
algorithm for this problem iterates the steps
\begin{subequations} \label{eq:admmStd}
\begin{align}
    \xbfhat_1  &\gets \argmin_{\xbf_1} L_\gamma(\xbf_1,\xbfhat_2,\sbf)  \\
    \xbfhat_2  &\gets \argmin_{\xbf_2} L_\gamma(\xbfhat_1,\xbf_2,\sbf) \\
    \sbf &\gets \sbf + \gamma(\xbfhat_1-\xbfhat_2) ,
\end{align}
\end{subequations}
where it becomes evident that $\gamma$ can also be interpreted as a step size.
The benefit of the ADMM method is that the minimizations involve
only one penalty, $f_1(\xbf)$ or $f_2(\xbf)$, at a time.  A classic result
\cite{BoydPCPE:09} shows that if the penalties $f_i(\xbf)$ are convex
(not necessarily smooth) and \eqref{eq:xmap} has a unique minima, then
the ADMM algorithm will converge to that minima.
Our next result relates MAP-GEC to ADMM.

\begin{theorem} \label{thm:admm} Consider GEC
(Algorithm~\ref{algo:GEC}) with the MAP estimation functions~\eqref{eq:gimap},
but with fixed second-order terms,
\beq \label{eq:gamadmm}
    \gamma_1 = \gamma_2 = \gamma, ~ \eta = \gamma_1+\gamma_2=2\gamma
\eeq
for some fixed scalar value $\gamma > 0$.
Define
\beq \label{eq:sadmm}
    \sbf_i^k = \gamma(\xbfhat_i^k- \rbf_i^k) .
\eeq
Then, the outputs of GEC satisfy
\begin{subequations} \label{eq:gecadmm}
\begin{align}
    \xbfhat_1^k &= \argmin_{\xbf_1} L_\gamma(\xbf_1,\xbfhat_2^k,\sbf_1^k)
         \label{eq:x1admm} \\
    \sbf_2^k &= \sbf_1^k + \gamma(\xbfhat_1^k-\xbfhat_2^k)
        \label{eq:s2admm} \\
    \xbfhat_2^{\kp1} &= \argmin_{\xbf_2} L_\gamma(\xbfhat_1^k,\xbf_2,\sbf_2^k)
         \label{eq:x2admm} \\
    \sbf_1^{\kp1} &= \sbf_2^k + \gamma(\xbfhat_1^k-\xbfhat_2^{\kp1})
        \label{eq:s1admm} .
\end{align}
\end{subequations}
\end{theorem}

Note that in the above description, we have been explicit about
the iteration number $k$ to be precise about the timing of the updates.
We see that a variant of ADMM can be interpreted as a special case
of GEC with particular, fixed step sizes.  This variant differs from
the standard ADMM updates by having two updates of the dual parameters
in each iteration.  Alternatively, we can think of GEC as a particular
variant of ADMM that uses an \emph{adaptive} step size. From our discussion
above, we know that the GEC algorithm can be interpreted as adapting the
step-size values $\gamma^k$ to match the local ``curvature" of the objective
function.
}

\iftoggle{conference}{
}{

\appendices

\section{Proof of Theorem~\ref{thm:mapFix}} \label{sec:mapFixPf}

Since $\xbfhat = \xbfhat_i = \gbf_i(\rbf_i,\gammabf_i)$, and $\gbf_i(\cdot,\cdot)$
is the MAP estimation function~\eqref{eq:gimap}, we have
\[
    \xbfhat = \argmin_{\xbf} \left[ f_i(\xbf) + \frac{1}{2}\|\xbf-\rbf_i\|^2_{\gammabf_i} \right].
\]
Hence,
\[
    \nabla f_i(\xbfhat) + \gammabf_i \vmult (\xbfhat-\rbf_i) = 0 ,
\]
where $\nabla f_i(\xbfhat)$ denotes the gradient of $f_i(\xbf)$ at $\xbf=\xbfhat$.
Summing over $i=1,2$ and applying \eqref{eq:xhatfix},
\[
    \nabla f_1(\xbfhat) + \nabla f_2(\xbfhat)
    = (\gammabf_1 \vmult \rbf_1 + \gammabf_2 \vmult \rbf_2)  -(\gammabf_1+\gammabf_2) \vmult \xbfhat
    = 0,
\]
which shows that $\xbfhat$ is a critical point of \eqref{eq:xmap}.

\section{Proof of Theorem~\ref{thm:curve}} \label{sec:curvePf}

Using \eqref{eq:gimapDeriv}, the fixed points of line~\ref{line:gamj} of Algorithm~\ref{algo:GEC} must satisfy
\beq \label{eq:gammappf}
    \gammabf_j = \mathbf{1}\vdiv\dbf(\Qbf_i) - \gammabf_i, \quad
    \Qbf_i =  (\Pbf_i + \Gammabf_i)^{-1}.
\eeq
Now, to prove part (a) of the theorem, suppose $\Pbf_i = \Diag(\pbf_i)$ for some vector $\pbf_i$.
Using \eqref{eq:gammappf} with the vector-valued diagonalization $\dbf(\Qbf) = \diag(\Qbf)$,
\begin{align*}
    \gammabf_j = \mathbf{1}\vdiv\diag\left[(\Pbf_i+\Gammabf_i)^{-1}\right] - \gammabf_i
    = \pbf_i+\gammabf_i-\gammabf_i
    = \pbf_i.
\end{align*}
Hence,
\[
    \etabf = \gammabf_1+\gammabf_2
    = \mathbf{1}\vdiv\diag\left[ (\Pbf_1+\Pbf_2)^{-1} \right] = \etabfhat .
\]

In part (b) of the theorem, we use uniform diagonalization \eqref{eq:Dunif}.
Recall that $\etabf$ has identical components, which we shall call $\eta$.
Likewise, $\gammabf_i$ are vectors with identical components $\gamma_i$.
Then from line~\ref{line:etai} of Algorithm~\ref{algo:GEC},
\[
    \eta^{-1}
    = \frac{\Tr(\Qbf_i)}{N}
    = \frac{\Tr\left((\Pbf_i + \gamma_i\Ibf)^{-1}\right)}{N}
    = S_{\Pbf_i}(-\gamma_i),
\]
which shows that $\gamma_i = -S_{\Pbf_i}^{-1}(\eta^{-1})$.
From line~\ref{line:gamj},
\[
    \gamma_j = \eta - \gamma_i = \eta + S_{\Pbf_i}^{-1}(\eta^{-1}) = R_{\Pbf_i}(-\eta^{-1}).
\]
Thus, using the freeness property \eqref{eq:free},
\begin{align*}
    \eta
    &= \gamma_1 + \gamma_2 = R_{\Pbf_1}(-\eta^{-1}) + R_{\Pbf_2}(-\eta^{-1}) \\
    &= R_{\Pbf_1+\Pbf_2}(-\eta^{-1}) = S_{\Pbf_1+\Pbf_2}^{-1}(\eta^{-1}) + \eta,
\end{align*}
and hence $S_{\Pbf_1+\Pbf_2}(0) =\eta^{-1}$. So,
\[
    \eta^{-1} = \frac{1}{N} \Tr\left((\Pbf_1+\Pbf_2)^{-1}\right) = \left[ \dbf\left((\Pbf_1+\Pbf_2)^{-1}\right) \right]_1
    =\widehat{\eta}^{-1}.
\]

\section{Proof of Theorem~\ref{thm:mmseFix}} \label{sec:mmseFixPf}

Corresponding to the objective function \eqref{eq:JBFE} with
moment matching constraints \eqref{eq:MMcon}, define the
Lagrangian,
\begin{align}
    \lefteqn{ L(b_1,b_2,q,\betabf,\gammabf) } \nonumber \\
    &:= J(b_1,b_2,q)
     - \sum_{i=1}^2 \betabf_i\tran\left[
        \Exp(\xbf|b_1) - \Exp(\xbf|q) \right]
    \nonumber\\&\quad
     + \sum_{i=1}^2 \gammabf_i\tran\left[
        \dbf(\Exp(\xbf\xbf\tran|b_1)) - \dbf(\Exp(\xbf\xbf\tran|q)) \right].
    \label{eq:LagBFE}
\end{align}
To show that $b_i$ and $q$ are stationary points of the constrained optimization,
we need to show that they satisfy the moment matching constraints \eqref{eq:MMcon}
and
\begin{align}
    b_i &= \argmin_{b_i} L(b_1,b_2,q,\betabf,\gammabf) \label{eq:bifixlag} \\
    q &= \argmax_{q} L(b_1,b_2,q,\betabf,\gammabf) \label{eq:qfixlag}
\end{align}

To prove \eqref{eq:bifixlag}, first observe that the Lagrangian \eqref{eq:LagBFE}
can be written as
\begin{align}
    L(b_1,b_2,q,\betabf,\gammabf)
    &= D(b_i\|e^{-f_i}) - \betabf_i\tran\Exp(\xbf|b_i)
    \nonumber\\&\quad
    + \frac{1}{2}\gammabf_i\tran\dbf(\Exp(\xbf\xbf\tran|b_i)) + \mbox{const},
    \label{eq:LagBFEbi}
\end{align}
where the constant terms do not depend on $b_i$.
Now, for the vector-valued diagonalization operator \eqref{eq:Dvec}, we have
\beq \label{eq:gamd}
    \gammabf_i\tran\dbf(\Exp(\xbf\xbf\tran))
    = \Exp\left[ \|\xbf\|^2_{\gammabf_i} \right].
\eeq
The same identity also holds when $\gammabf_i$ is a constant vector
and $\dbf(\cdot)$ is the uniform diagonalization operator \eqref{eq:Dunif}.
Substituting \eqref{eq:gamd} into \eqref{eq:LagBFEbi}, we obtain
\begin{align}
    \lefteqn{ L(b_1,b_2,q,\betabf,\gammabf) }\nonumber \\
        &= D(b_i\|e^{-f_i}) - \betabf_i\tran\Exp(\xbf|b_i)
         + \frac{1}{2}\Exp\left[ \|\xbf\|^2_{\gammabf_i} \right] + \mbox{const} \nonumber \\
        &\stackrel{(a)}{=}
            D(b_i\|e^{-f_i}) + \frac{1}{2}\Exp\left[ \|\xbf-\rbf_i\|^2_{\gammabf_i} \right]
        + \mbox{const}  \nonumber \\
        &\stackrel{(b)}{=}
            -H(b_i) + \Exp\left[ f_i(\xbf) + \frac{1}{2}\|\xbf-\rbf_i\|^2_{\gammabf_i} \right]
        + \mbox{const} \nonumber \\
        &\stackrel{(c)}{=} D\left( b_i \| p_i(\cdot|\rbf_i,\gammabf_i) \right) + \mbox{const},
        \label{eq:LagBFEbi2}
\end{align}
where in step (a) we used the fact that $\betabf_i = \gammabf_i \vmult \rbf_i$, and in steps
(b) and (c) we used the definitions of KL divergence and $p_i(\cdot)$ from \eqref{eq:pirgam}.
Thus, the minimization in \eqref{eq:bifixlag} yields \eqref{eq:bifix}.

The maximization over $q$ in \eqref{eq:qfixlag} is computed
similarly.  Removing the terms that do not depend on $q$,
\begin{align}
    \lefteqn{ L(b_1,b_2,q,\betabf,\gammabf) }\nonumber\\
        &= H(q) + \sum_{i=1}^2 \betabf_i\tran\Exp(\xbf|q)
        - \frac{1}{2}\sum_{i=1}^2\Exp\left[ \|\xbf\|^2_{\gammabf_i} \right] + \mbox{const} \\
        &\stackrel{(a)}{=}
            H(q) + (\etabf\vmult\xbfhat)\tran\Exp(\xbf|q)  
        - \frac{1}{2}\Exp\left[ \|\xbf\|^2_{\etabf} \right]   + \mbox{const}  \nonumber \\
        &\stackrel{(b)}{=}
            H(q) - \frac{1}{2}\Exp\left[ \|\xbf-\xbfhat\|^2_{\etabf} \right]   + \mbox{const}  \nonumber \\
        &\stackrel{(c)}{=} -D(q\|\widehat{q} \,) + \mbox{const},
        \label{eq:LagBFEqi}
\end{align}
where step (a) uses the facts that $\gammabf_1+\gammabf_2=\etabf$ and
\[
    \betabf_1 + \betabf_2 = \gammabf_1 \vmult \rbf_1+\gammabf_2 \vmult \rbf_2 = \etabf \vmult \xbfhat,
\]
step (b) follows by completing the square, and step (c) uses the density
\[
    \widehat{q}(\xbf) \propto \exp\left[ \frac{1}{2}\|\xbf-\xbfhat\|^2_{\etabf} \right].
\]
Hence, the maximum in \eqref{eq:qfixlag} is given by $q=\widehat{q}$, which matches \eqref{eq:qfix}.

Also, from lines~\ref{line:xi} and \ref{line:etai} of Algorithm~\ref{algo:GEC}
and \eqref{eq:gimmse} and \eqref{eq:gimmseDeriv},
\[
    \xbfhat = \Exp(\xbf|b_i), \quad \mathbf{1}\vdiv\etabf = \dbf(\Cov(\xbf|b_i)).
\]
Since $\widehat{q}$ is Gaussian, its mean and covariance matrix are
\[
    \Exp(\xbf|q) = \xbfhat, \quad \Cov(\xbf|q) = \Diag(\mathbf{1}\vdiv\etabf).
\]
This proves that the densities satisfy the moment matching constraints \eqref{eq:MMcon}.

\section{Proof of Theorem~\ref{thm:convFO} } \label{sec:convFOPf}

Let $\gammabf = \gammabf_1+\gammabf_2$ and $\Gammabf = \Diag(\gammabf)$.
Define the scaled variables $\vbf_i = \Gammabf^{-1/2}\betabf_i$.
Since $\gammabf > 0$, it suffices to prove the convergence of $\vbf_i$.
We can rewrite the update \eqref{eq:betaFO} as
\beq \label{eq:vup}
    \vbf_j \gets F_i(\vbf_i) :=
    \Gammabf^{1/2}\gt_i(\Gammabf^{1/2}\vbf_i,\gammabf_i) - \vbf_i.
\eeq
So, we have that the updates are given by the recursion
\[
    \vbf_2 = F_1(\vbf_2) = F_1(F_2(\vbf_2)).
\]
If $\Jbf_i(\vbf_i) = \partial F_i(\vbf_i)/\partial \vbf_i$ is the Jacobian of transformation,
then, by the chain rule, the Jacobian of the composition is $\partial (F_1\circ F_2)/\partial \vbf_1
= \Jbf_1 \Jbf_2$.  A standard contraction mapping result \cite{Vidyasagar:78} shows that if,
for some $\rho < 1$,
\[
    \|\Jbf_1(\vbf_1) \Jbf_2(\vbf_2)\| \leq \rho ,
\]
then $\vbf_i$ converges linearly to a unique fixed point.

So, we need to characterize the norms of the Jacobians.
First, the Jacobian of the update function in \eqref{eq:vup} is
\beq
    \Jbf_i(\vbf_i) = \Gammabf^{1/2}\Qbftilde_i\Gammabf^{1/2} - \Ibf,
    \quad
    \Qbftilde_i = \frac{\partial \gt_i(\betabf_i,\gammabf_i)}{\partial \betabf_i}
\eeq
Using Assumption~\ref{as:firstOrder}, $\Qbftilde_i$ is symmetric and
hence so is $\Jbf_i$.  Also, using \eqref{eq:Qderivbnd} and the fact that
$\etabf=\gammabf_1+\gammabf_2$,
\begin{align} \label{eq:Jibndup}
    \Jbf_i(\vbf_i)
    &\leq \Diag\left[ (\gammabf_1+\gammabf_2)\vdiv(c_{i2}\mathbf{1} + \gammabf_i) - \mathbf{1} \right] \\
    &= \Diag\left[ (\gammabf_j-c_{i2}\mathbf{1})\vdiv(\gammabf_i+c_{i2}\mathbf{1}) \right],
\end{align}
for $(i,j)=(1,2)$ or $(2,1)$.
Similarly,
\begin{align} \label{eq:Jibnddown}
    \Jbf_i(\vbf_i)
    &\geq \Diag\left[ (\gammabf_1+\gammabf_2)\vdiv(c_{i1}\mathbf{1} + \gammabf_i) - \mathbf{1} \right] \\
    &= \Diag\left[ (\gammabf_j-c_{i1}\mathbf{1})\vdiv(\gammabf_i+c_{i1}\mathbf{1}) \right] .
\end{align}
Thus, the matrix absolute value of $\Jbf_i$ (i.e.\ from the spectral theorem,
not componentwise) satisfies
\[
    |\Jbf_i| \leq \Diag\left[
        \frac{|\gammabf_j-\qbf_i|}{|\gammabf_i+\qbf_i|}\right],
\]
where $\qbf_i$ has components $q_{in}=c_{i1}$ or $c_{i2}$.
Thus,
\[
    |\Jbf_1||\Jbf_2| \leq \Diag\left[
    \frac{|\gammabf_2-\qbf_1||\gammabf_1-\qbf_2|}{
    |\gammabf_1+\qbf_1||\gammabf_2+\qbf_2|}\right].
\]
Now for each component $n$,
\[
    |\gamma_{2n}-q_{1n}||\gamma_{1n}-q_{2n}|
    < |\gamma_{1n}+q_{1n}||\gamma_{2n}+q_{2n}|,
\]
since all the terms are positive.  Since this is true for all $n$,
there must exist a $\rho \in [0,1)$ such that,
\[
    |\Jbf_1||\Jbf_2| \leq \rho \Ibf.
\]
Note that the value of $\rho$ can be selected independently of $\vbf_i$.
Therefore, the norm of the Jacobian product satisfies
\[
    \| \Jbf_1\Jbf_2 \| \leq \| |\Jbf_1||\Jbf_2| \| \leq \rho
\]
for all $\vbf_i$.
Hence, the mapping is a contraction.

\section{Proof of Theorem~\ref{thm:convMAP}} \label{sec:convMAPPf}
We start by proving part (a).
First, recall that $(\rbf_1^0,\gammabf_1^0)$ were constructed by running
Algorithm~\ref{algo:GEC} to convergence under \emph{fixed} $\gammabf_i^0>0$
and $\etabf_i^0=\gammabf_1^0+\gammabf_2^0$ for $i=1,2$.
Theorem~\ref{thm:mapFix} studied this recursion and showed that it
yields final values ``$\rbf_i^0$'' of $\rbf_i$ for which the corresponding
estimates $\xbfhat_i^0:=\gbf_i(\rbf_i^0,\gammabf_i^0)$ satisfy
$\xbfhat_1^0=\xbfhat_2^0=\xbfhat$, where $\xbfhat$ is a local minima of \eqref{eq:xmap}.
Since we have assumed that $f_i(\xbf)$ are strictly convex, $\xbfhat$ is
the unique global minimizer.
Theorem~\ref{thm:convMAP} then considers what happens when Algorithm~\ref{algo:GEC}
is run from the initialization $(\rbf_1^0,\gammabf_1^0)$ \emph{without} holding
$\gammabf_i,\etabf_i$ fixed.

We now show that $\xbfhat_1=\xbfhat_2=\xbfhat$ for all iterations.
We prove this by induction.
Suppose that $\xbfhat_i=\xbfhat$ (which we know holds for $i=1$ during the first iteration due to the construction of the initialization $\xbfhat_1^0$).
Since
$\xbfhat_i = \gbf_i(\rbf_i,\gammabf_i)$ with $\gbf_i(\cdot,\cdot)$ being the minimizer
in \eqref{eq:gimap}, $\xbfhat=\xbfhat_i$ must satisfy the first-order condition
\beq \label{eq:xideriv}
    \nabla f_i(\xbfhat) + \gammabf_i \vmult (\xbfhat - \rbf_i) = 0.
\eeq
Therefore,
\begin{align*}
    \MoveEqLeft \nabla f_j(\xbfhat) + \gammabf_j \vmult (\xbfhat - \rbf_j) \\
    & \stackrel{(a)}{=} -\nabla f_i(\xbfhat) + \gammabf_j \vmult (\xbfhat - \rbf_j) \\
    &\stackrel{(b)}{=} -\nabla f_i(\xbfhat) + (\etabf_i-\gammabf_i) \vmult \xbfhat - \gammabf_j \vmult \rbf_j \\
    &\stackrel{(c)}{=} -\gammabf_i \vmult \rbf_i + \etabf_i \vmult \xbfhat - \gammabf_j \vmult \rbf_j
    \stackrel{(d)}{=} 0,
\end{align*}
where (a) follows from the fact that $\xbfhat$ is the minimizer of
\eqref{eq:xmap} and so $\nabla f_i(\xbfhat) + \nabla f_j(\xbfhat) = 0$;
(b) follows from line~\ref{line:gamj} of Algorithm~\ref{algo:GEC};
(c) follows from the induction hypothesis \eqref{eq:xideriv};  and
(d) follows from line~\ref{line:rj}.
Hence $\xbfhat$ satisfies the first-order minimization conditions for
$\xbfhat_j = \gbf_j(\rbf_j,\gammabf_j)$, so $\xbfhat_j = \xbfhat$.
This proves part (a).

We now turn to part (b).  We will prove this part for the vector-valued
diagonalization operator \eqref{eq:Dvec}.  The proof for the
uniform diagonalization operator \eqref{eq:Dunif} is similar, but easier.
Now, from part (a), we can assume that $\xbfhat_i=\xbfhat$
for all iterations.
From \eqref{eq:gimapDeriv}, we see that $\Qbf_i$ in line~\ref{line:Qi} is given by
\beq \label{eq:Qimappf}
    \Qbf_i = (\Pbf_i + \Gammabf_i)^{-1}, \quad \Pbf_i = \Hessian{f_i}(\xbfhat),
\eeq
where $\Hessian{f_i}(\xbfhat)$ is the Hessian.
Hence, from line~\ref{line:gamj}, the updates of the second-order terms are given by
\beq \label{eq:gamjmap}
    \gammabf_j \gets G^i(\gammabf_i)
    := \mathbf{1}\vdiv\diag\left( (\Pbf_i + \Gammabf_i)^{-1}\right) - \gammabf_i.
\eeq
Note that, since the penalty functions $f_i(\xbf)$ satisfy the assumptions in
Lemma~\ref{lem:convFn}, we have $\Pbf_i := \Hessian{f_i}(\xbfhat) > 0$.
Also, we can write the update on $\gammabf_1$ as
\beq \label{eq:gamcomp}
     \gammabf_1 \gets G^2 \circ G^1(\gammabf_1),
\eeq
where $G^2 \circ G^1$ is the composition map.

Now, to analyze this update, consider a general map of the form
\beq \label{eq:gammap}
    G(\gammabf) := \mathbf{1}\vdiv\diag\left( (\Pbf + \Gammabf)^{-1} \right) - \gammabf,
\eeq
where, as before, we use the notation $\Gammabf = \diag(\gammabf)$ and $\Pbf > 0$.
We prove four properties of a map of this form:  For any $\gammabf > 0$,
\begin{enumerate}[(i)]
\item Non-negative:  $G(\gammabf) \geq 0$;
\item Non-decreasing:  $G(\gammabf') \geq G(\gammabf)$ when $\gammabf' > \gammabf$;
\item Sub-multiplicative:  For any $\alpha > 1$, $G(\alpha \gammabf) \leq \alpha G(\gammabf)$;
\item Bounded:  There exists a $M$ such that $G(\gammabf) \leq M$.
\end{enumerate}

If we can show that $G(\gammabf)$ in \eqref{eq:gammap}
satisfies these four properties, then so does $G^i(\gammabf_i)$ in \eqref{eq:gamjmap}
and hence their composition $G^1 \circ G^2$.  It is then shown in \cite{yates:95} that
the update in \eqref{eq:gamcomp} must converge to a unique fixed point.  So, we must simply
prove that $G(\gammabf)$ in \eqref{eq:gammap} satisfies properties (i) to (iv).

For property (i), observe that the components of $G(\gammabf)$ are given by
\[
    G_i(\gammabf) = \left[ (\Pbf + \Gammabf)^{-1} \right]^{-1}_{ii} - \gamma_i.
\]
Since $\Pbf \geq 0$,
\[
    G_i(\gammabf) \geq \left[ \Gammabf^{-1} \right]^{-1}_{ii} - \gamma_i = 0.
\]
So $G_i(\gammabf) \geq 0$.  This proves property (i).
Next, we prove that it is increasing.
Let
\[
    \Sbf = (\Pbf + \Gammabf)^{-1},
\]
so that we can write $G_i(\gammabf)$ as
\[
    G_i(\gammabf) = S_{ii}^{-1} - \gamma_i.
\]
Then
\begin{align*}
    \lefteqn{ \frac{\partial G_i(\gammabf)}{\partial \gamma_j}
        = -S_{ii}^{-2} \frac{\partial}{\partial \gamma_j} \ebf_i\tran\Sbf\ebf_i - \delta_{i-j} } \\
        &= -S_{ii}^{-2} \ebf_i\tran \left[ \frac{\partial}{\partial \gamma_j} (\Pbf+\Gammabf)^{-1} \right] \ebf_i - \delta_{i-j} \\
        &= S_{ii}^{-2} \ebf_i\tran (\Pbf+\Gammabf)^{-1} \left[ \frac{\partial}{\partial \gamma_j} (\Pbf+\Gammabf)\right] (\Pbf+\Gammabf)^{-1} \ebf_i - \delta_{i-j} \\
        &= S_{ii}^{-2} \ebf_i\tran \Sbf \left[ \ebf_j \ebf_j\tran \right] \Sbf \ebf_i - \delta_{i-j} \\
        &= S_{ij}^2 / S_{ii}^2 - \delta_{i-j} . \\
\end{align*}
For $i\neq j$, we see that
\[
     \frac{\partial G_i(\gammabf)}{\partial \gamma_j} = \frac{S_{ij}^2}{S_{ii}^2} \geq  0,
\]
and for $i = j$, we have
\[
    \frac{\partial G_i(\gammabf)}{\partial \gamma_j} = \frac{S_{ii}^2}{S_{ii}^2} - 1 = 0.
\]
Hence, the function is non-decreasing, which proves property (ii).
Next, we need to show that is sub-multiplicative.
Suppose $\alpha > 1$.  Then
\begin{align*}
    G(\alpha \gammabf)
    &= \mathbf{1}\vdiv\diag\left( (\Pbf + \alpha \Gammabf)^{-1} \right) - \alpha \gammabf \\
    &\leq \mathbf{1}\vdiv\diag\left( (\alpha \Qbf + \alpha \Gammabf)^{-1} \right) - \alpha \gammabf
    = \alpha G(\gammabf),
\end{align*}
which proves property (iii).
Lastly, we need to show it is bounded above.  First notice that we can write
\begin{align*}
G(\gammabf)
&= \diag(\Ibf - \Gammabf(\Pbf+\Gammabf)^{-1}) \vdiv \diag( (\Pbf+\Gammabf)^{-1}) \\
&= \diag(\Pbf(\Pbf+\Gammabf)^{-1}) \vdiv \diag( (\Pbf+\Gammabf)^{-1}) \\
&= \diag(\Pbf(\Pbf+\Gammabf)^{-1}\Gammabf) \vdiv \diag( (\Pbf+\Gammabf)^{-1}\Gammabf) \\
&= \diag((\Gammabf^{-1}+\Pbf^{-1})^{-1}) \vdiv \diag( (\Gammabf^{-1}\Pbf+\Ibf)^{-1}) .
\end{align*}
Then as $\gammabf \arr \infty$, we have
\begin{align*}
    G(\gammabf)
    &\rightarrow \diag((\Pbf^{-1})^{-1}) \vdiv \diag( \Ibf^{-1}) = \diag(\Pbf) .
\end{align*}
If $\Pbf$ is the Hessian of a penalty function that satisfies the assumptions
in Lemma~\ref{lem:convFn}, then $\diag(\Pbf)$ is be bounded from above,
implying that $G(\gammabf)$ is also bounded from above.
This proves properties (i) to (iv) above.

\section{Proof of Theorem \ref{thm:fixr}} \label{sec:fixrPf}

For the penalty $f_2(\xbf)$ in \eqref{eq:fawgn},
the belief estimate $b_2(\xbf)$ in \eqref{eq:bifix} is Gaussian with covariance matrix
\[
    \Cov(\xbf|b_2) = \left[ \gamma_w\Abf^*\Abf + \gamma_2\Ibf \right]^{-1}
    = \left[\Ybf + \gamma_2\Ibf\right]^{-1},
\]
where $\Ybf = \gamma_w\Abf\tran\Abf$.
From \eqref{eq:gimapDeriv} and line~\ref{line:etai} in Algorithm~\ref{algo:GEC},
$\Qbf_2 = \Cov(\xbf|b_2)$. Under the uniform diagonalization operator \eqref{eq:Dunif},
$\etabf$ has identical components $\eta$ such that
\[
    \eta^{-1} = \frac{1}{N}\Tr(\Qbf_2) = S_{\Ybf}(-\gamma_2),
\]
where $S_{\Ybf}(\omega)$ is the Stieltjes transform \eqref{eq:stieltjes}.
Hence,
\[
    \gamma_2 = -S_{\Ybf}^{-1}(\eta^{-1}).
\]
Also, $\gammabf_i$ has identical components $\gamma_i$ for each $i=1,2$.
Thus from line~\ref{line:gamj} of Algorithm~\ref{algo:GEC},
\[
    \gamma_1 = \eta - \gamma_2 =   \eta + S_{\Ybf}^{-1}(\eta^{-1})
    = R_{\Ybf}(-\eta^{-1}).
\]
This proves the first equation in \eqref{eq:fixr}.  The second equation
is exactly \eqref{eq:etafixr}.

\section{Proof of Theorem ~\ref{thm:admm}} \label{sec:admmPf}
To prove \eqref{eq:x1admm},
\begin{align}
    \xbfhat_1^k &\stackrel{(a)}{=} \argmin_{\xbf_1} f_1(\xbf_1) + \frac{\gamma}{2}\|\xbf_1-\rbf_1^k\|^2 \nonumber \\
    &\stackrel{(b)}{=}   \argmin_{\xbf_1} f_1(\xbf_1) +
    \frac{\gamma}{2}\|\xbf_1-\xbfhat_2^k-\gamma^{-1}\sbf_1^k\|^2 \nonumber \\
    &\stackrel{(c)}{=} \argmin_{\xbf_1} f_1(\xbf_1) + \frac{\gamma}{2}\|\xbf_1-\xbfhat_2^k\|^2
    +(\sbf_1^k)\tran\xbf_1 \nonumber \\
    &\stackrel{(d)}{=} \argmin_{\xbf_1} L(\xbf_1,\xbfhat_2,\sbf_1^k),
\end{align}
where (a) follows from $\xbfhat_1^k = \gbf_1(\rbf_1^k,\gamma)$ and
the definition of the MAP estimation function in \eqref{eq:gimap};
(b) follows from the definition of $\sbf_1^k$ in \eqref{eq:sadmm};
(c) follows by expanding the squares and (d) follows from the augmented
Lagrangian in \eqref{eq:Ladmm}.  The proof of \eqref{eq:x2admm} is similar.

To prove \eqref{eq:s1admm}, first observe from the update of $\rbf_j$ in
line \eqref{line:rj} of Algorithm~\ref{algo:GEC} and the fact that
$\eta = 2\gamma$, we have
\beq \label{eq:r2admm}
    \gamma \rbf_2^k = 2\gamma \xbfhat_1^k - \gamma \rbf_1^k.
\eeq
Therefore,
\begin{align}
    \MoveEqLeft \sbf_2^k \stackrel{(a)}{=}
        \gamma(\rbf_2^k-\xbfhat_1^k) \nonumber \\
        & \stackrel{(b)}{=}  \gamma (\xbfhat_1^k-\rbf_1^k)
        = \sbf_1^k + \gamma(\xbfhat_1^k-\xbfhat_2^k),
\end{align}
where (a) follows from the definition of $\sbf_2^k$ in \eqref{eq:sadmm};
(b) follows from \eqref{eq:r2admm} and (c) follows from
\eqref{eq:sadmm}.  The proof of \eqref{eq:s2admm} is similar.

}

\bibliographystyle{IEEEtran}
\bibliography{bibl}
\end{document}